\documentclass[5p]{elsarticle}

\usepackage{tikz}
\usepackage{amsthm}
\usetikzlibrary{shapes,arrows,positioning,calc}
\tikzset{
block/.style = {draw, fill=white, rectangle, minimum height=2.5em, minimum width=3em},
tmp/.style = {coordinate},
sum/.style= {draw, fill=white, circle, node distance=1cm},
input/.style = {coordinate},
output/.style= {coordinate},
pinstyle/.style = {pin edge={to-,thin,black}}}

\newtheorem{prop}{Proposition}

\usepackage[subtle]{savetrees}

\newtheorem{remark}{Remark}
\usepackage{amsmath}
\usepackage{amssymb}
\usepackage{epstopdf}
\usepackage[font=small,labelfont=bf]{caption}
\usetikzlibrary{patterns}
\usetikzlibrary{decorations.markings}
\usepackage{romannum}
\usepackage{cases}
\usepackage{color}
\usepackage{algorithm,algorithmic}
\usepackage{cancel}
\usepackage{setspace}
\usepackage{arydshln}
\renewcommand{\arraystretch}{0.7}

\begin{document}

\begin{frontmatter}

\title{\vspace{-1.2cm} \huge \bf Preview Reference Governors: \\A Constraint Management Technique \\ for Systems With Preview Information}

\author[rvt]{Yudan Liu}
\ead{yliu38@uvm.edu}

\author[rvt]{Hamid. R. Ossareh}
\ead{Hamid.Ossareh@uvm.edu}
\address[rvt]{Department of Electrical and Biomedical Engineering, University of Vermont, Burlington, VT, USA, 05405}


\fntext[]{Funding for this work was generously provided by Ford under the URP Award 2016-8004R, and NASA under grant 80NSSC20M0213.}

\begin{abstract}
This paper presents a constraint management strategy based on Scalar Reference Governors (SRG) to enforce output, state, and control constraints while taking into account the preview information of the reference and/or disturbances signals. The strategy, referred to as the Preview Reference Governor (PRG), can outperform SRG while maintaining the highly-attractive computational benefits of SRG. However, as it is shown, the performance of PRG may suffer if large preview horizons are used. An extension of PRG, referred to as Multi-horizon PRG, is proposed to remedy this issue. 
Quantitative comparisons between SRG, PRG, and Multi-horizon PRG on a one-link  robot arm example are  presented to illustrate their performance and computation time. Furthermore, extensions of PRG are presented to handle systems with disturbance preview and multi-input systems. The robustness of PRG to parametric uncertainties and inaccurate preview information is also explored. 
\end{abstract}


\end{frontmatter}


\begin{abstract}

This paper presents a constraint management  strategy based on Scalar Reference Governors (SRG) to enforce output, state, and control constraints while taking into account the preview information of the reference and/or disturbances signals. The strategy, referred to as the Preview Reference Governor (PRG), can outperform SRG while maintaining the highly-attractive computational benefits of SRG. However, as it is shown, the performance of PRG may suffer if large preview horizons are used. An extension of PRG, referred to as Multi-horizon PRG, is proposed to remedy this issue. 
Quantitative comparisons between SRG, PRG, and Multi-horizon PRG on a one-link arm robot example are  presented to illustrate their performance and computation time. The robustness of PRG to disturbance previews, parametric uncertainties, and inaccurate preview information is also explored.


\end{abstract}

\section{INTRODUCTION} \label{sec: intro}
Preview control has been a subject of study in the 
 field of control theory for many decades.
The main idea behind preview control is to incorporate known or estimated information on the future values of the disturbances or references (i.e., ``preview" information) in the computation of the current control command. Such preview may be computed from models or may be available from measurements. For example, in a wind turbine control application, preview information on wind velocity may be available from measurements taken elsewhere in the wind farm. Incorporating this information in the calculation of the control command can result in improved system performance  \cite{koerber2013combined}.

Common methods for preview control are $H_\infty$ preview control \cite{takaba2003tutorial}, LQR preview control \cite{farooq2005path}, and mixed $H_2$-$H_\infty$ method, just to mention a few (see \cite{birla2015optimal} for a comprehensive review). These methods, however, are not able to enforce pointwise-in-time state and control constraints, which is important to ensure safe and efficient system operation. A preview control method that can, in fact, enforce these constraints is Model Predictive Control (MPC) \cite{bemporad2002model,morari1999model,camacho2013model}, which is an optimization-based method that can naturally incorporate preview information in the cost function  \cite{gohrle2012active}. Some applications of preview-MPC can be found in
active suspension \cite{Gohrle_2014},
swing leg trajectories of biped walking robots \cite{shimmyo2012biped}, and wind turbines \cite{koerber2013combined}. Since MPC is computationally demanding, its applicability has been limited to systems with slow or low-order dynamics, or systems controlled by fast processors     \cite{Cairano_2012}. 
A relatively new constraint management strategy, which alleviates the above computational challenges of MPC, is the Reference Governor, also referred to as the Scalar Reference Governor (SRG) \cite{gilbert1995discrete,Garone_2017,Kolmanovsky_2014,liu2020decoupled, garone2015explicit,ossareh2020reference}. It is an add-on mechanism that enforces pointwise-in-time state and control constraints by governing, whenever is required, the {\it reference} signal of a pre-stabilized closed-loop system (see Figure~\ref{fig:RGblock}). However, standard SRG uses only the value of the reference signal at the current time and is unable to take preview information into account. This paper fills this gap and presents a novel reference governor-based solution, referred to as the Preview Reference Governor (PRG), to enforce  the constraints  while incorporating the preview information of the reference and disturbance signals, which yields superior transient performance as compared to SRG. 

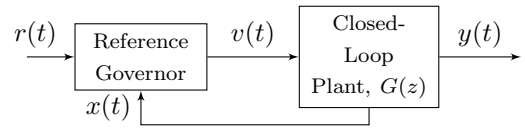
\begin{figure}
\centering
\begin{tikzpicture}[auto, node distance=1.5cm,>=latex']
\node [input, name=rinput] (rinput) {};
\node [block, right of=rinput,text width=1.5cm,align=center] (controller) {{\footnotesize Reference Governor}};
\node [block, right of=controller,node distance=3cm,text width=1.6cm,align=center] (system)
{{\footnotesize Closed-Loop Plant, $G(z)$}};
\node [output, right of=system, node distance=2cm] (output) {};
\node [tmp, below of=controller,node distance=0.9cm] (tmp1){$s$};
\draw [->] (rinput) -- node{\hspace{-0.4cm}$r(t)$} (controller);
\draw [->] (controller) -- node [name=v]{$v(t)$}(system);
\draw [->] (system) -- node [name=y] {$y(t)$}(output);
\draw [->] (system) |- (tmp1)-| node[pos=0.75] {$x(t)$} (controller);
\end{tikzpicture}
\caption{Scalar reference governor block diagram. The signals are as follows: $y(t)$ is the constrained output, $r(t)$ is the reference, $v(t)$ is the governed reference, and $x(t)$ is the system state (measured or estimated).} \label{fig:RGblock}
 \label{fig: Governor scheme}
\end{figure}

To explain PRG, we first summarize the main ideas behind SRG. 
The SRG employs the so-called Maximal Admissible set (MAS) \cite{Gilbert_1991}, which is defined as the set of all initial conditions and inputs that ensure constraint satisfaction for all times. This set is computed offline. In real-time, SRG computes an optimal $v(t)$ (See Figure \ref{fig:RGblock}) to maintain the system state inside the MAS and, thus, enforce the constraints. This is achieved by solving, at every time step, a simple linear program, whose solution can be computed explicitly. 
Now reconsider the closed-loop plant $G(z)$ in Figure~\ref{fig: Governor scheme}, but assume that the preview information of the reference signal is available.  More specifically, ${r(t)}, {r(t+1)}, \ldots, {r(t+N)}$, are known to the controller at time $t$, where $N$ is the ``preview horizon". We initially assume that $G(z)$ has a single input (we will extend the theory to multi-input systems in Section \ref{sec: PRG for multi-inputs}). Similar to SRG, the goal of PRG is to select $v(t)$ as close as possible to $r(t)$ (to ensure that the tracking performance does not suffer) such that the output constraints are satisfied for all times. 
However, the PRG must take into consideration the preview information of $r(t)$ in order to improve the tracking performance as compared to SRG. To achieve these goals, we first lift  $r(t)$ and $v(t)$ from $\mathbb{R}$ to $\mathbb{R}^{(N+1)}$ in order to describe them over the entire preview horizon. Next, the system $G$ is represented based on the lifted input, and a new MAS is calculated based on this new representation. Finally, PRG is formulated as an extension of SRG, where a new optimization problem is solved based on the new MAS and a new update law is used to compute $v(t)$. A high-level block diagram of the PRG is shown in Figure \ref{fig: preview Governor scheme}, where the availability of the preview information of the reference is explicitly displayed.

As we show in this paper, PRG guarantees constraint satisfaction, recursively feasibility, and closed-loop stability. However, it may calculate conservative inputs under certain conditions, especially for long preview horizons. To overcome this issue and further improve performance, we present an extension of PRG. The extension, which we refer to as Multi-horizon PRG (abbreviated as Multi-$N$ PRG), is based on solving multiple PRGs with different preview horizons and optimally fusing the outputs of the PRGs.

We next consider the case where the system is affected by disturbances, whose preview information is available (in addition to the reference preview). The PRG solution for this case involves formulating a new MAS, wherein the disturbance preview appears explicitly as parameters of the MAS. To guarantee constraint satisfaction, the MAS is robustified to the worst-case realization of disturbances beyond the preview horizon. 

In practice, systems are affected by uncertainties in the model as well as the preview information. Therefore, we also consider systems under  parametric uncertainties and inaccurate preview information (in either reference or disturbance). We propose a ``robust" PRG formulation to robustify the design against these uncertainties. 

Finally, we propose two solutions to handle multi-input systems. The first solution comprises the same ideas as in PRG for single-input systems, except that we lift $r(t)$ and $v(t)$ from $\mathbb{R}^m$ to $\mathbb{R}^{(N+1)m}$, where $m$ is the number of inputs. However, as will be shown by an example, this solution might lead to a conservative response since a single optimization variable is used to govern all input channels. To overcome this shortcoming, we propose another solution, which combines the Decoupled Reference Governor scheme \cite{liu2018decoupled} and PRG theory. Detailed information will be explained in Section \ref{sec: PRG for multi-inputs}.

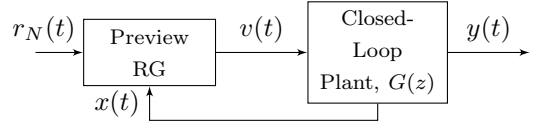
\begin{figure}
\centering
\begin{tikzpicture}[auto, node distance=1.5cm,>=latex']
\node [input, name=rinput] (rinput) {};
\node [block, right of=rinput,text width=1.5cm,align=center] (controller) {{\footnotesize Preview RG}};
\node [block, right of=controller,node distance=3cm,text width=1.6cm,align=center] (system)
{{\footnotesize Closed-Loop Plant, $G(z)$}};
\node [output, right of=system, node distance=2cm] (output) {};
\node [tmp, below of=controller,node distance=0.9cm] (tmp1){$s$};
\draw [->] (rinput) -- node{\hspace{-0.4cm}$r_N(t)$} (controller);
\draw [->] (controller) -- node [name=v]{$v(t)$}(system);
\draw [->] (system) -- node [name=y] {$y(t)$}(output);
\draw [->] (system) |- (tmp1)-| node[pos=0.75] {$x(t)$} (controller);
\end{tikzpicture}
\caption{Preview Reference Governor block diagram. $r_N(t)$ represents the lifted reference over the preview horizon, i.e., $r_N(t)=(r(t),\ldots,r(t+N))$.} 
 \label{fig: preview Governor scheme}
\end{figure}

In summary, the original contributions of this work are:
\begin{itemize}
    \item A novel RG-based constraint management scheme with preview capabilities, namely the PRG, which incorporates preview information of the references into the reference governor framework and is more computationally efficient than existing methods such as MPC or Command Governors (CG) \cite{Garone_2017};
    
    \item An extension of PRG (Multi-$N$ PRG) to further improve the performance of PRG;
    
    \item Analysis of recursively feasibility, closed-loop stability, and convergence under constant inputs.
    
    \item Comparison of the computational footprint and performance of these schemes.
    
    \item Extensions of PRG to  systems with disturbance preview, parametric uncertainties, inaccurate preview information, and multi-input systems.
\end{itemize}



The following notations are used in this paper. $\mathbb{Z}$ and $\mathbb{Z}_+$ denote the set of all integers and non-negative integers, respectively. The identity matrix with dimension $p\times p$ is denoted by $I_p$. The variable $t\in \mathbb{Z}_+$ is the discrete time. For vectors $x$ and $y$, $x \leq y$ is to be interpreted component-wise. A zero matrix with dimension $p \times q$ is denoted as $0_{p \times q}$.




\section{Review of Maximal Admissible Sets And Scalar Reference Governors}\label{sec: Review O_inf}
\subsection{Maximal Admissible Set (MAS)} \label{sec: review_maximal}
Consider Figure \ref{fig:RGblock}, where $G$ represents the closed-loop system whose dynamics are described by:
\begin{equation}\label{eq:system linear}
\begin{aligned}
x(t+1)&=Ax(t)+Bv(t)\\
y(t)&=Cx(t)+Dv(t) 
\end{aligned}
\end{equation}
where  $x(t)\in \mathbb{R}^n$ is the state vector, $v(t)\in \mathbb{R}^m$ is the input, and $y(t)\in \mathbb{R}^p$ is the constrained output vector. Since $G$ represents the closed-loop system with a stabilizing controller, it is assumed that $G$ is asymptotically stable. Over the output, the following polyhedral constraint is imposed:
\begin{equation}\label{eq: constraint set}
    y(t)\in \mathbb{Y}:=\{y:Sy\leq s\}
\end{equation}
The MAS,
denoted by $O_\infty$, is the set of all initial states and constant control inputs that
satisfy \eqref{eq: constraint set} for all future time steps:
\begin{equation}\label{eq:definition of O_inf}
O_\infty := \{(x_0,v_0)\in \mathbb{R}^{n+m}: x(0)=x_0,v(t)=v_0,
y(t)\in \mathbb{Y}, \forall t\in \mathbb{Z}_+ \}
\end{equation}
As seen in \eqref{eq:definition of O_inf}, to construct MAS, $v(t)=v$ is held constant for all $t\in \mathbb{Z}_+$. Using this assumption, the evolution of the output $y(t)$ can be expressed explicitly as a function of the initial state $x(0)=x_0$ and the constant input $v_0$ as:
\begin{equation}\label{eq: model recursion}
\begin{aligned}
y(t)&=CA^tx_0+C(I-A)^{-1}(I-A^t)B v_0 +D v_0 .
\end{aligned}
\end{equation}

Therefore, the MAS in \eqref{eq:definition of O_inf} can be characterized by a polytope defined by  an infinite number of inequalities, i.e.,
\begin{equation}\label{eq:definition of O_inf polytope}
\begin{aligned}
&O_\infty := \{(x_0,v_0)\in \mathbb{R}^{n+m}:
H_x x_0 + H_v v_0 \leq h\}
\end{aligned}
\end{equation}
where $H_x=[CA^t], H_v=[C(I-A)^{-1}(I-A^t)B +D] $, and $h=[s^\top,s^\top,\ldots]^\top$. Reference \cite{Gilbert_1991}  shows that under mild assumptions on $C$ and $A$, it
is possible to make the set \eqref{eq:definition of O_inf polytope} finitely determined (i.e., the matrices $H_x$, $H_v$ and $h$ are finite dimensional) by constraining the steady-state
value of $y$, denoted by $y(\infty)$, to the interior of the constraint set, i.e.,
\begin{equation}\label{eq: y_infty}
    y(\infty):=(C(I-A)^{-1}B+D)v\in (1-\epsilon) \mathbb{Y}
\end{equation}
where $\epsilon>0$ is a small number. By introducing \eqref{eq: y_infty} in \eqref{eq:definition of O_inf polytope}, there exists a finite prediction  time $t^*$, where the inequalities corresponding to all future prediction times ($t>t^*$) are redundant. This yields an inner approximation of \eqref{eq:definition of O_inf polytope}, denoted by $\bar{O}_\infty$, which is finitely determined.  

\subsection{Scalar Reference Governors (SRG)}

The SRG calculates $v(t)$ based on $\bar{O}_\infty$ by introducing an internal state, whose dynamics are governed by:
\begin{equation}\label{eq: RG formulation}
v(t)=v(t-1)+\kappa(r(t)-v(t-1)),
\end{equation}
where $\kappa$ is found by solving the following linear program (LP):
\begin{equation}\label{eq:RG_kappa}
\begin{aligned}
\underset{\kappa\in [0,1]}{\text{maximize}}\quad
& \mathrm{\kappa} \\
\text{s.t.} \quad
&  (x(t),v(t)) \in \bar{O}_\infty\\
& v(t)=v(t-1)+\kappa(r(t)-v(t-1))\\
\end{aligned}
\end{equation}
From (\ref{eq: RG formulation}), the SRG computes an input $v(t)$ based on the convex combination of the previous admissible input value (i.e., $v(t-1)$) and the current reference $r(t)$ in order to guarantee constraint satisfaction. Also,  if the pair \mbox{$(x(t),v(t-1))$} is admissible, then $(x(t+1),v(t))$ is also admissible, implying recursively feasibility of the SRG \cite{Garone_2017}.

\section{Preview Reference Governors (PRG) for single-input systems} \label{sec: preview RG}
In this section, the PRG for single-input systems is  introduced and analyzed. In addition, PRG is compared with SRG using a numerical example to show that it can significantly improve the closed-loop system performance  while enforcing the constraints. An extension of PRG to multi-input systems will be explained in Section \ref{sec: PRG for multi-inputs}. As mentioned in the Introduction, PRG is based on lifting the reference (i.e., $r(t)$) and the governed reference (i.e., $v(t)$) from $\mathbb{R}$ to $\mathbb{R}^{(N+1)}$ and representing the system using the lifted input. Based on the new representation, a modified maximal admissible set is characterized and a new formulation of reference governor is proposed, as explained in detail below.
 
Consider the system shown in Figure \ref{fig: preview Governor scheme}. Assume $r(t)$, $r(t+1)$, $\ldots$, $r(t+N)$ are available at time $t$, where $r(t) \in \mathbb{R}$ is the current value of the setpoint and $r(t+1),\ldots,r(t+N) \in \mathbb{R}$ are the preview information, and $N\geq 0$ represents the preview horizon. Note that $N=0$ corresponds to the case where no preview information is available. We now define the lifted signals $r_N(t)\in \mathbb{R}^{(N+1)}$ (shown in Figure \ref{fig: preview Governor scheme}) and $v_N(t)\in \mathbb{R}^{(N+1)}$, as follows:
\begin{equation}\label{eq: def of v_bar}
\begin{aligned}
 r_N(t)=(r(t),  \ldots, r(t+N)), \quad v_N(t)=(
v(t), \ldots, v(t+N))
\end{aligned}
\end{equation}
Using the lifted signals, $G(z)$ can be equivalently expressed as:
\begin{equation}\label{augment state-space system}
\begin{aligned}
    x(t+1) = A x(t) + 
    \underbrace{\begin{bmatrix}
    B & 0 & \ldots &0
    \end{bmatrix}}_{\widetilde{B}} v_N(t),\\
    y(t) = C x(t) + 
    \underbrace{\begin{bmatrix}
    D & 0 & \ldots &0
    \end{bmatrix}}_{\widetilde{D}} v_N(t)
\end{aligned}
\end{equation}

We now describe the Maximal Admissible Set (MAS) for this system. Recall from \eqref{eq:definition of O_inf} that in order to characterize MAS in the SRG framework, it is assumed that $v(t)$ is held constant for all time. This will ensure that the optimization problem \eqref{eq:RG_kappa} will always have a feasible solution, namely ${\kappa=0}$. In order to extend these ideas to PRG, we assume that $v(t)$ may vary within the preview horizon, but is held constant beyond the preview horizon. 
Therefore, the dynamics of $v_N(t)$ are chosen to be:
\begin{equation} \label{eq: A_bar}
\begin{aligned}
v_N(t+1)= \bar{A} v_N(t)
    \end{aligned}
    \end{equation}
with initial condition $v_N(0)=v_{N_0}$, where $\bar{A}$ is defined by: 
\begin{equation} \label{eq: A_bar_def}
\bar{A}=
\left[
	\begin{array}{c:c}
	\smash{\underbrace{\begin{matrix}
    0 & 1  &\ldots &0\\
    \vdots &\vdots  &\ddots & \vdots\\
    0 & 0  & \ldots &1 \\
    0 & 0  &\ldots &0\\ \hdashline
    0 & 0  &\ldots &0 \\
    \end{matrix}}_{N}}&
    \begin{matrix}
    0 \\[0.1mm] \vdots \\[0.1mm] 0 \\[0.1mm] 1  \\[1.5mm]\hdashline  1 
    \end{matrix}
    \end{array}
	\right]
	\end{equation}
\vspace{0.1cm}
\par \noindent This choice of $\Bar{A}$, together with the definition of $v_N(t)$ in \eqref{eq: def of v_bar}, enforce that for all $t \geq N$, $v(t)= v(N)$. With these dynamics, the new MAS is defined by:
\begin{equation}\label{eq:O_inf for preview}
\begin{aligned}
O_\infty^N:=&\{(x_0,v_{N_0})\in \mathbb{R}^{n+(N+1)}: x(0)=x_0,v_N(0)=v_{N_0}\\
& v_N(t+1)=\Bar{A} v_N(t), y(t)\in \mathbb{Y}, \forall t \in \mathbb{Z}_+ \}
\end{aligned}
\end{equation}
Similar to \eqref{eq: y_infty},  a  finitely determined inner approximation of this set can be computed by tightening the steady-state constraint (to prove finite determinism, note that after $N$ time steps, $v(t)$  converges to a constant, which reduces the problem to that in the standard MAS theory). In the rest of this paper, with an abuse of notation, we use $O_\infty^N$ to denote the finitely determined inner approximation (instead of using $\bar{O}_\infty^N$).

With the new MAS defined, we are now ready to present the PRG formulation. Recall that the SRG computes $v(t)$ using the update law in  \eqref{eq: RG formulation}, where $\kappa$ is obtained by solving the linear program (LP) in \eqref{eq:RG_kappa}. Note from  \eqref{eq: RG formulation} that  $v(t) \in \mathbb{R}$ can be regarded as the internal state of the SRG (i.e., the SRG strategy consists of a single internal state). To extend these ideas to PRG, we introduce $(N+1)$ new states in the formulation, where the state update law is as follows: 
\begin{equation}\label{eq:kapa}
v_N(t)=\Bar{A} v_N(t-1)+\kappa(r_N(t)-\Bar{A} v_N(t-1))
\end{equation}
where $\kappa$ is the solution of the following linear program:
\begin{equation}\label{eq: LP to compute kappa}
\begin{aligned}
&\underset{\kappa\in [0,1]}{\text{maximize}}
& & \mathrm{\kappa} \\
& \hspace{10pt} \text{s.t.}
& & v_N(t)=\Bar{A} v_N(t-1)+\kappa(r_N(t)-\Bar{A} v_N(t-1))\\
&&&(x(t),v_N(t)) \in O_\infty^N
\end{aligned}
\end{equation}
An explicit algorithm, similar to the one in SRG \cite{liu2018decoupled}, can be developed to solve this LP efficiently (see Section \ref{sec:computation} for details). The PRG solves the above LP at every time step to compute $\kappa$, updates $v_N(t)$ using \eqref{eq:kapa}, and applies the first element of $v_N(t)$ to the system $G$.  
Note that the variables in the LP  in \eqref{eq: LP to compute kappa} at time step $t$ are $v_{N}(t-1)$, $r_N(t)$, and $x(t)$. 

\begin{remark} 
When $N=0$, PRG reduces to SRG because $\bar{A}$ turns into the scalar with value $1$ 
in this case. Therefore, PRG is a proper extension of SRG.
\end{remark}

Several important properties of the PRG are described in the following proposition.
\begin{prop}\label{prop:1}
The PRG formulation is recursively feasible, bounded-input and bounded-output stable (BIBO), and for a constant $r$, $v(t)$ converges. 
\end{prop}

\begin{proof}
To show recursively feasibility, consider the update law \eqref{eq:kapa}. As can be seen, $\kappa=0$ implies that $v_N(t)=\bar{A}v_N(t-1)$, which matches the  dynamic of $v_N$ that is assumed in $O_\infty^N$ (see \eqref{eq:O_inf for preview}). Positive invariance  of $O_\infty^N$ implies that if the pair $(x(t-1),v_N(t-1))$ is admissible, then $(x(t),v_N(t-1))$ is also admissible. In conclusion, $\kappa=0$ is always a feasible solution to the LP in \eqref{eq:RG_kappa}, implying recursively feasibility of the PRG. As for BIBO stability, recall that $v_N(t)$ is the convex combination of $\bar{A}v_N(t-1)$ and the current reference $r_N(t)$. Thus, if $r(t)$ is bounded, then so is $v_N(t)$. This, together with the asymptotic stability of $G$, implies BIBO stability of the system.  To prove the convergence property, assume $r(t) = r, \forall t\in \mathbb{Z}_+$. From \eqref{eq:kapa}, it can be shown that every element in $v_N(t)$ is monotonic $\forall t \geq N$ and bounded by $r$. Thus, $v_N(t)$ must converge to a limit. Note that, since the closed-loop system $G$ is asymptotically stable, $x(t)$ and $y(t)$ would also converge.
\end{proof}

\begin{remark}
Note that the definition of $v_N$ in \eqref{eq: def of v_bar} is consistent with the delay structure in the dynamics \eqref{eq: A_bar}--\eqref{eq: A_bar_def}. However, the definition in \eqref{eq: def of v_bar} may appear inconsistent with the update law in \eqref{eq:kapa} because the delay structure vanishes when ${\kappa \neq 0}$. We remark that this is not an error. The definition in \eqref{eq: def of v_bar} and the dynamics in \eqref{eq: A_bar}--\eqref{eq: A_bar_def} are only introduced to construct $O_\infty^N$. For real-time implementation, $v_N$ should not be thought of as defined by \eqref{eq: def of v_bar}, but instead should be thought of as the PRG's internal states. 
\end{remark}

\begin{remark}
For the case where the state  of $G(z)$ (i.e., $x(t)$ shown in Figure \ref{fig: preview Governor scheme}) is not measured, standard Luenberger observers can be designed to estimate the state. 
\end{remark}

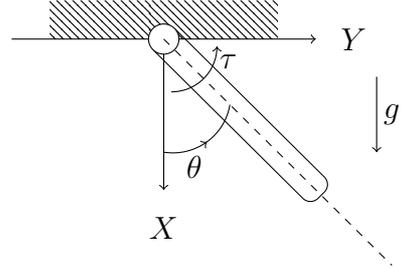
\begin{figure}
\centering
\begin{tikzpicture}[arrowmark/.style 2 args={decoration={markings,mark=at position #1 with \arrow{#2}}}]
	 \draw[->](-2,0)-- (2,0);
	 \draw[->](0,0)--(0,-2);
      \fill[pattern=north west lines, pattern color=black] (-1.5,0) rectangle (1.5,0.5);
    \node (a) at (0, 0) {};
    \node (p) at (2,-2) {};
  \draw[black,rounded corners]
    let \p1=($(a)!-1mm!(p)$),
        \p2=($(p)!-1mm!(a)$),
        \p3=($(\p1)!2mm!90:(\p2)$),
        \p4=($(\p1)!2mm!-90:(\p2)$),
        \p5=($(\p2)!2mm!90:(\p1)$),
        \p6=($(\p2)!2mm!-90:(\p1)$)
    in
    (\p3) -- (\p4)-- (\p5) -- (\p6) -- cycle;
 
	 \draw (2.5,0)node
      [color=black,font=\fontsize{12}{12}\selectfont]{$Y$};
      
	 \draw (0,-2.5)node
      [color=black,font=\fontsize{12}{12}\selectfont]{$X$};
      
     \draw[black,fill=white] (0,0) circle (0.2);
     
     \draw[dashed](0,0)-- (3,-3);
     
	 \draw (0.4,-1.7)node
      [color=black,font=\fontsize{12}{12}\selectfont]{$\theta$};
      
      \draw[
    postaction={decorate},
    arrowmark={0.5}{>},
    ](0,-1.5)to[bend right=45](0.87,-0.87);
      \draw[
    postaction={decorate},
    arrowmark={1}{>},
    ](0.1,-0.7)to[bend right=45](0.7,-0.1);
    
	 \draw (0.85,-0.3)node
      [color=black,font=\fontsize{12}{12}\selectfont]{$\tau$};
      
      \draw[->](2.8,-0.5)--(2.8,-1.5);
	 \draw (3,-1)node
      [color=black,font=\fontsize{12}{12}\selectfont]{$g$};   
\end{tikzpicture}
\caption{One-link arm.}
\label{fig: One-link arm}
\end{figure}

We now illustrate the PRG using a numerical simulation of a one-link arm robot, shown in Figure~\ref{fig: One-link arm}. A state-space model of the arm is given by:
\begin{equation}\label{eq: model for robot}
\begin{aligned}
&\begin{bmatrix}\dot{x_1}\\ \dot{x_2}
\end{bmatrix}=
\begin{bmatrix}0 & 1\\ 
-14.7 & 0
\end{bmatrix}
\begin{bmatrix}x_1\\ x_2
\end{bmatrix} 
+  \begin{bmatrix}0\\ 
3
\end{bmatrix}\tau, \quad y=\begin{bmatrix}1 & 0
\end{bmatrix} \begin{bmatrix}x_1\\ x_2
\end{bmatrix} 
 \end{aligned}
\end{equation}
where $x_1 \triangleq \theta$, $x_2 \triangleq\dot{\theta}$ are the states and $\tau$ (i.e., external torque) is the control input. 
For this example, we assume that both states are  measured (if not, an observer can be designed). 


In order to design a controller and consequently implement the PRG, the system model~\eqref{eq: model for robot} is first discretized at $T_s=0.01s$.  Then, a state-feedback controller with a pre-compensator is applied  to ensure that the output, $\theta$, tracks the desired angle, $v$, perfectly, that is:
$
\tau=66.67v-
61.77 x_1 - 9.64 x_2.
$
This results in the closed-loop system $G$ shown in Figure~\ref{fig: preview Governor scheme}. 
We simulate a trajectory-following maneuver, wherein we use the PRG to ensure that the output, $\theta$, remains within $[-45^{\circ}, 45^{\circ}]$. The preview horizon, $N$, is chosen to be $25$ (i.e., $0.25$ seconds). We assume that the preview information is available from the given pre-set trajectory.

The comparison between the performance of PRG and SRG is shown in Figure \ref{fig: previewRG on robot}. 
As can be seen, $v(t)$ given by PRG is closer to $r(t)$ (less conservative) when $t\in[0.3,0.5]$. This is because when $t=0.3$, the PRG has the future information that the reference would drop down to $0$ in future $25$ time steps so that it allows a larger $v(t)$ than SRG. For the same reason, $v(t)$ given by PRG is less conservative when $t \in [0.7,0.9]$.

\begin{figure}[t]
\centering
\includegraphics[scale=0.28]{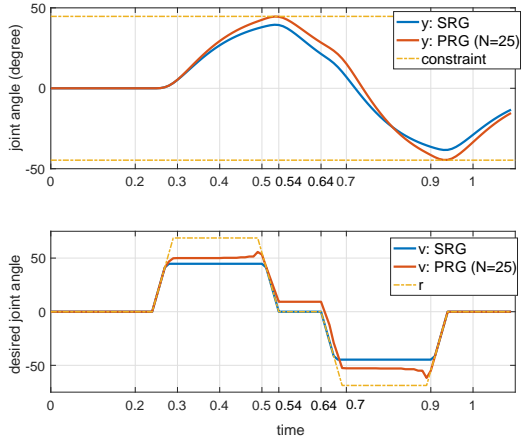}
\caption{Comparison of PRG and SRG. The blue  lines represent the results of SRG and the red lines refer to the results of PRG. The top plot shows the outputs and the bottom plot shows the control inputs and the setpoint.}
\label{fig: previewRG on robot}
\end{figure}

Figure \ref{fig: previewRG on robot} shows a limitation of PRG as well. Specifically, it can be seen that when $t\in[0.54,0.64]$,  $v(t)$ given by PRG  can not reach $r(t)=0$, even though $0$ is an admissible input. To explain the root cause of this behavior, note that when $t=0.54$, the preview information available to the PRG is that $r(t)$ drops down to $-69^{\circ} $ at $t=0.7$ and stays constant  afterwards (the increase of $r(t)$ back to 0 at $t=0.9$ is beyond the preview horizon and unavailable to PRG). To enforce the lower constraint for $t>0.7$, the PRG calculates a $\kappa$ smaller than 1. However, since $\kappa$ affects all elements of $v_N(t)$ (see  \eqref{eq:kapa}), this leads to a $v(t)$ that is different from $r(t)$ when $t \in [0.54,0.64]$, leading to a conservative solution.

Of course, the above limitation of PRG can  be addressed by decreasing the preview horizon $N$. However, if $N$ is too short, the tracking performance of the system might not be improved significantly as compared with SRG. In order to have the superior system performance while  addressing the above limitation of PRG, we provide an extension of PRG  in the following section.

\section{Multi-Horizon PRG (Multi-$N$ PRG)}\label{sec: multi-horizon PRG}
To overcome the limitation of PRG described in the previous section and improve performance, this section introduces the Multi-Horizon Preview Reference Governor (Multi-$N$ PRG). As the name suggests, instead of having a single preview horizon, multiple preview horizons are considered. For each horizon, a MAS is characterized, and multiple PRGs are solved at each time step, one for each MAS. A technical challenge for this approach is that $v_N(t)$'s for different preview horizons would have different dimensions and different interpretations. In addition, a practical challenge is that storing multiple MASs may lead to a significant increase in the memory requirements. Our novel solution overcomes these challenge by unifying the $v_N$'s so that only one MAS is required, and fusing the PRG solutions in a special way to guarantee constraint satisfaction and recursively feasibility.

To begin, consider the following set of $q$ preview horizons: $N_1 < N_2<\ldots<N_q$. Let  $O_\infty^{N_i}$ be defined as in \eqref{eq:O_inf for preview}, i.e., the MAS of the lifted system with preview horizon $N_i$. We first show that there exists a simple relationship between $O_\infty^{N_i}$, $i=1,\ldots, q-1$, and $O_\infty^{N_q}$. Recall from  \eqref{eq:O_inf for preview} and \eqref{eq: def of v_bar} that in order to construct $O_\infty^{N_i}$, it is assumed that $v(t)$ is held constant beyond the preview horizon (i.e., after ${t=N_i}$). Similarly, to construct $O_\infty^{N_q}$, it is assumed that $v(t)$ is held constant after ${t=N_q}$. Thus, $O_\infty^{N_i}$ and $O_\infty^{N_q}$ have the following relationship:
$$
(x,v_{N_i}) \in O_\infty^{N_i} \Leftrightarrow
(x, [v_{N_i}^T, \underbrace{v_{N_i}(N_i+1)^T, \ldots, v_{N_i}(N_i+1)^T}_{N_q-N_i \text{    terms}}]^T) \in O_\infty^{N_q}
$$
where $\Leftrightarrow$ denotes the bidirectional implication and ${v_{N_i}(N_i+1)}$ is the last element in $v_{N_i}$. Therefore, given this  relationship, we only need to compute and store  $O_\infty^{N_q}$.

We now introduce the Multi-$N$ PRG, which  solves $q$ PRGs at each time step, one for each $N_i$. All PRGs use the same MAS, namely $O_\infty^{N_q}$, as justified above. To provide more details, recall from Section \ref{sec: preview RG} that PRG contains an internal state with an update law given by \eqref{eq:kapa}. Similarly, we introduce an internal state, denoted here by $\widetilde{v}\in\mathbb{R}^{(N_q+1)}$, for the Multi-$N$ PRG. PRG $i$ in the Multi-$N$ PRG framework assumes the update law:
\begin{equation}\label{eq:updatemulti}
\widetilde{v}(t)=M_i (\bar{A}\widetilde{v}(t-1)+\kappa_i(r_{N_q}(t)-\bar{A}\widetilde{v}(t-1)) \end{equation}
where $\bar{A}$ is defined the same as \eqref{eq: A_bar_def} but with $N=N_q$, and  $r_{N_q}(t)$ is the lifted version of $r$ at time $t$ (defined as in \eqref{eq: def of v_bar}). The matrix $M_i$ is introduced above to force the control inputs beyond the preview horizon $N_i$ to be constant (this is to maintain consistency with the fact that PRG $i$ has a preview horizon of length $N_i$). The matrix $M_i$ that achieves this is:
\[M_i = \begin{bmatrix}
     I_{(N_i+1)} & 0_{(N_i+1) \times (N_q-N_i)}\\
     \widetilde{I}_{(N_q-N_i)\times (N_i+1)} 
     & 0_{(N_q-N_i) \times (N_q-N_i)}
    \end{bmatrix}\]
where $\widetilde{I}_{(N_q-N_i)\times (N_i+1)}$ is a matrix with zeros everywhere except the rightmost columns, which are given by $[1,1,\ldots,1]^T$. In the update law \eqref{eq:updatemulti}, for  PRG $i$, the scalar $\kappa_i$ is computed by solving the following LP:
\begin{equation}\label{eq: kappa for multi kappa}
\begin{aligned}
&\underset{\kappa_i \in [0,1]}{\text{maximize}}
\quad  \mathrm{\kappa_i} \\
& \hspace{10pt} \text{s.t.}
\quad \widetilde{v}(t)= M_i ( \bar{A}\widetilde{v}(t-1)+\kappa_i(r_{N_q}(t)-\bar{A}\widetilde{v}(t-1))\\
& \quad \quad\quad (x(t),  \widetilde{v}(t)) \in O_\infty^{N_q}
\end{aligned}
\end{equation}

\noindent To fuse the solutions of the PRGs,  the maximum value among $\{ \kappa_i \}, i=1,\ldots,q$ is selected, denoted by $\kappa_{i^*}$. The index that corresponding to $\kappa_{i^*}$ is denoted by $i^*$. Then, the final update law of Multi-$N$ PRG is obtained by fusing the PRGs as follows:
\[ \widetilde{v}(t)=M_{i^*} (\bar{A}\widetilde{v}(t-1)+\kappa_{i^*}(r_{N_q}(t)-\bar{A}\widetilde{v}(t-1))\] 
This formulation maintains recursively feasibility. Indeed, suppose PRG $i^*$ is the PRG that computes $\kappa_{i^*}$ at time step $t$. Then, the same PRG can calculate a feasible solution at time $t+1$, due to the recursively feasibility of the PRG formulation, which was proven in Proposition \ref{prop:1}.

Note that not all PRGs in the Multi-$N$ PRG formulation will have feasible solutions at every time step. If a feasible solution to these PRGs does not exist, $\kappa_i$'s for these PRGs are set to be $0$. Note also that even though multiple preview horizons are considered in Multi-$N$ PRG, only $O_\infty^{N_q}$ is required to be computed and stored offline, implying that if $N_q$ for Multi-$N$ PRG is equal to $N$ for PRG, then the memory requirements of the two formulations are the same. We will study the processing requirements in the next section. 

The simulation results  of Multi-$N$ PRG on the one-link arm robot  example are shown in Figure \ref{fig: multi kappa previewRG}. For comparison, the results of the implementation of PRG with $N=25$ is also provided (this is the same as Figure \ref{fig: previewRG on robot}). For the sake of illustration, we consider the extreme case where the Multi-$N$ PRG uses all preview horizons between 0 and 25; i.e., $q=26$ and $N_1 = 0$, $N_2 = 1$, \ldots , $N_{26} = 25$.

\begin{figure}[t]
\centering	
\includegraphics[scale=0.28]{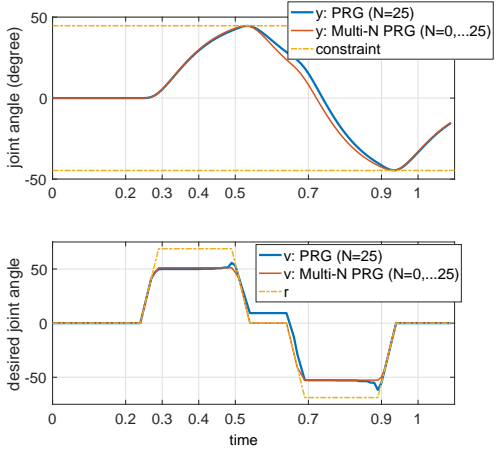}
\caption{Comparison of Multi-$N$ PRG and PRG with $N=25$. 
}
\label{fig: multi kappa previewRG}
\end{figure}

From Figure \ref{fig: multi kappa previewRG}, the outputs for both Multi-$N$ PRG and PRG satisfy the constraints, as expected. However, from the bottom plot of Figure \ref{fig: multi kappa previewRG}, it can be seen that when ${t\in [0.54,0.64]}$, $v(t)$ given by Multi-$N$ PRG reaches $r(t)$  while $v(t)$ computed by PRG is above $r(t)$. The reason for this behavior is that when $t\in [0.54,0.64]$, the PRG corresponding to $N_1=0$ computes $\kappa=1$, which leads to $v(t)=r(t)$. Note that the actual performance improvement (as seen on the output plot) is not large for this specific example but it could be large in general. 


Next, we will discuss the impact of the value of $q$ in the Multi-$N$ PRG on the system performance. Consider two extreme cases for the sake of illustration: first, $N_1$ and $N_2$ are chosen to be $0$ and $100$ (i.e., $q=2$); second, $N_1,\ldots,N_{100}$ are chosen to be all numbers from $0$ to $100$ (i.e., $q=101$). The reason  $N_q=100$ is selected is that the preview horizon in this case is $1s$ (recall the sample time is $T_s=0.01$), which is longer than the variation of the reference in our example, and can clearly show the difference of the system performance between the two cases. The comparison between the two cases implemented on the one-link arm robot example is shown in Figure~\ref{fig: comparison of q}, which demonstrates performance improvements when a larger $q$ is used. Generally, when the preview horizon is longer than the expected variations of the reference, it is better to use the Multi-$N$ PRG formulation with a larger $q$. However, this leads to an increase in the computational burden, which is discussed next. 


\begin{figure}[t]
\centering	
\includegraphics[scale=0.28]{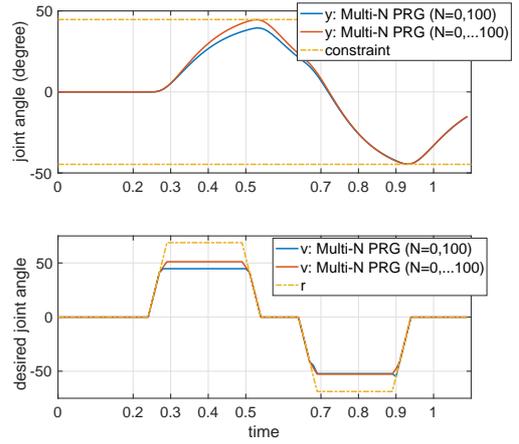}
\caption{Comparison of Multi-$N$ PRG with $q=2$ and $q=101$. 
}
\label{fig: comparison of q}
\end{figure}

\section{Computational considerations}\label{sec:computation}
In this section, we provide a comparison between the computation time of PRG, Multi-$N$ PRG, and standard SRG by  simulating the one-link arm robot example with all three methods. Recall that SRG and PRG require the solution to one linear program (LP) at each time step, while Multi-$N$ PRG requires the solutions to $q$ LP's ($q=26$ for our example). However, we do not use generic LP solvers to solve them. Instead, we use the Algorithm presented in \cite{liu2018decoupled} to solve the LPs in SRG and Algorithm \ref{Algorithm:algorithm1} below to solve the LPs in PRG. The LPs in  Multi-$N$ PRG can be solved using a similar algorithm. The notation in Algorithm \ref{Algorithm:algorithm1} is as follows. It is assumed that $O_\infty^{N}$ is finitely determined and given by polytopes of the form \eqref{eq:definition of O_inf polytope}. In addition, $j^*$ denotes the number of rows of $H_x$, $H_v$, and $h$.

\begin{algorithm}
 \caption{Custom Explicit PRG Algorithm}
 \begin{algorithmic}[1]
 \STATE let $a=H_{v}(r_N(t)-\bar{A}v_N(t-1))$
 \STATE let $b=h-H_{x}x(t)-H_{v}\bar{A}v_N(t-1)$
 \STATE set $\kappa=1$
  \FOR {$i = 1$ to $j^*$}
  \IF {$b(i) < 0$}
  \STATE $\kappa=0$
  \ELSE
    \IF {$a(i) > 0$}
    \STATE $\kappa=\min(\kappa,b(i)/a(i))$
  \ENDIF
  \ENDIF
  \ENDFOR
  \STATE $\kappa=\max(\kappa,0)$
 \end{algorithmic} 
 \label{Algorithm:algorithm1}
 \end{algorithm}

We simulate the one-link arm robot example with all three methods, namely SRG, PRG, and Multi-$N$ PRG. All simulations were performed for 150 time steps in Matlab on an Apple Macbook Pro with $1.4$ GHz Intel Core i5 processor and $8$ GB memory. In order to eliminate the effects of background processes running on the computer, each of the above experiments were run $10$ times and the averages are calculated.  We calculate the per-timestep averages and maximums of each of the three methods. The results are shown in Table  \ref{table: computation time}. 
As can be seen, PRG runs two orders of magnitude slower than SRG (because the matrices that arise in the computations are larger). The Multi-$N$ PRG is slower by one  order of magnitude (because $q$ PRGs are solved at each time step). 

Finally, to provide a comparison of these computation times with those of other existing preview control methods, we simulate the one-link arm robot example with the PRG replaced by a Command Governor (CG). Similar to MPC, a CG solves a Quadratic Program (QP) at each time step to directly optimize over $v(t)$. In addition, similar to MPC, it can incorporate preview information into the cost function. In brief, our CG solves the following QP:
\begin{align*}
\text{minimize}\quad
& (r_N(t)-v_N(t))^TQ(r_N(t)-v_N(t)) \\
\text{s.t.} \quad
&  (x(t),v_N(t)) \in O_\infty^N
\end{align*}
where $Q=Q^T>0$. We simulate the system with above CG. The QP above is solved at every time step using explicit Multi-Parametric Quadratic Programming (MPQP), which is introduced in \cite{tondel2003algorithm}. The computation times of CG are shown in Table \ref{table: computation time}. 
 As can be seen, CG is one order of magnitude slower than Multi-$N$ PRG. Note that we also implemented online QP for the CG with $N=25$, provided by MPT3 Toolbox in Matlab and Gurobi, and found that the computation times for both cases were longer than that of MPQP.

\begin{table}[h]
\renewcommand{\arraystretch}{1.3}
\vspace{0.3cm}
\caption{Comparison of the computation time between SRG, PRG, Multi-$N$ PRG, and CG for one-link arm robot example}
\label{table: computation time}
\centering
\renewcommand{\arraystretch}{0.9}
\begin{tabular}{c|c|c}
\hline
& \bfseries SRG & \bfseries PRG ($N=25$)\\
\hline
avg & $ 6.8\times 10^{-7}$s & $3.1 \times 10^{-5}$s \\
\hline
max & $9.2 \times 10^{-7}$s & $4.74 \times 10^{-5}$s\\
\hline
\hline

& \bfseries Multi-$N$ PRG $(N_q=25)$ & \bfseries CG $(N=25)$\\
\hline
avg & $6.54 \times 10^{-4}$s & $6.3 \times 10^{-3}$ \\
\hline
max & $8.5 \times 10^{-4}$s & $0.0162$ \\
\hline
\end{tabular}
\vspace{-0.25cm}
\end{table}

\section{Robust Preview Reference Governor}\label{sec: robust PRG}
In this section, PRG is extended to systems with disturbance preview, as well as systems with parametric uncertainties and preview information uncertainties. 

\subsection{Preview Reference Governor with Disturbance Preview}
\label{sec: PRG with prevew disturbance}
In previous sections, we considered systems in which  the preview information of the reference signal is available. However, there are situations wherein preview information of disturbance signals may be available as well. For example, in printing systems, the effect of paper (i.e., the disturbance) on the heating or charging systems are known with some preview, since the timing of the paper leaving the printing tray is precisely controlled \cite{ching2010modeling}. Hence, in this section, we consider systems where preview information of disturbances is available within a given preview horizon. For simplicity, we do not consider preview for the reference signal, though the results can be combined with those of the previous sections to consider preview on both references and disturbances.

Consider a system with additive bounded disturbance given by:
\begin{equation}\label{eq: LTI with disturbance}
\begin{aligned}
& x(t+1)=A x(t) + B v(t) + B_w w(t) \\
& y(t)= C x(t) + D v(t) + D_w w(t) 
\end{aligned}
\end{equation}
where $w \in \mathbb{W}$ and $\mathbb{W}$ is a compact polytopic set with origin in its interior. 
Essentially, we incorporate the disturbance preview into the RG formulation as follows: the maximal admissible set (MAS) is characterized as a function of the current state, input, and the known disturbances within the preview horizon. The set is then shrunk to account for the worst-case realization of the unknown disturbance beyond the preview horizon. Specifically, the robust MAS for systems with  disturbance preview can be defined as:
\begin{equation}\label{eq: O_inf for preview dist}
\begin{aligned} O_\infty^w=& \{(x_0,v_0,w_0,\ldots,w_N): x(0)=x_0,v(0)=v_0, 
 w(i)=w_i, \\ &0\leq i \leq N, y(t) \in  \mathbb{Y}, \forall t \in \mathbb{Z}_+, 
 \forall w(j) \in \mathbb{W}, j> N  \} 
\end{aligned}
\end{equation}
To compute this set, define $\mathbb{Y}_t=\mathbb{Y}$ for $t=0,\ldots,N$, $\mathbb{Y}_{N+1}=\mathbb{Y} \sim D_w \mathbb{W}$, and $\mathbb{Y}_{t+1}=\mathbb{Y}_t \sim CA^{t-N-1} B_w \mathbb{W}$ for $t \geq N+1$, where $\sim$ represents the Pontryagin subtraction \cite{gilbert1995discrete}. Then, the condition $y(t) \in  \mathbb{Y}$ in \eqref{eq: O_inf for preview dist} can be characterized equivalently by:
$$
CA^t x_0+\left(\sum\limits_{k=0}^{t-1}CA^{k}B+D\right)v_0 + h_d(t) \in \mathbb{Y}_t
$$
where $h_d(t)$ is defined as:
$$
h_d(t)=
\begin{cases}
\sum\limits_{k=0}^{t-1}CA^{t-1-k}B_ww(k)+D_ww(t) \quad \text{if $t \leq N$}\\
\sum\limits_{k=0}^{N}CA^{t-1-k}B_w w(k) \quad \quad \quad\quad\quad\,\,\, \text{if $t > N$}\\
\end{cases} 
$$

Note that if $N=0$, PRG with previewed bounded disturbance reduces to robust SRG with unknown bounded disturbance. In summary,  by shrinking the MAS to take the worst case disturbance into consideration, robust PRG guarantees constraints satisfaction for all values of disturbance beyond the preview horizon.

\begin{prop}\label{prop: robustPRG to disturbances}
The robust PRG formulation to  disturbance previews is BIBO stable, recursively feasible, and for a constant $r$, $v(t)$ converges.
\end{prop}

\begin{proof}
Similar to the proof for Proposition \ref{prop:1}
\end{proof}

We now illustrate the above idea with the one-link arm robot example. The disturbance is assumed to come from the joint torque, i.e. $B_w=B$ and $D_w=0$. We also assume that the disturbance satisfies $w \in \mathbb{W}:=[-0.1,0.1]$. For the purpose of simulations, the disturbance is generated randomly and uniformly from the interval $[-0.1, 0.1]$. The results of PRG with the disturbance as the preview information are shown in Figure~\ref{fig: previewRG disturbance}. Two preview horizons are chosen, namely  $N=20$ and $N=50$ in order to illustrate the relationship between the system performance and the length of the preview horizon. For comparison, Figure~\ref{fig: previewRG disturbance} also shows the response of robust SRG with unknown bounded disturbance (i.e., no preview).

\begin{figure}[t]
\centering	
\includegraphics[scale=0.28]{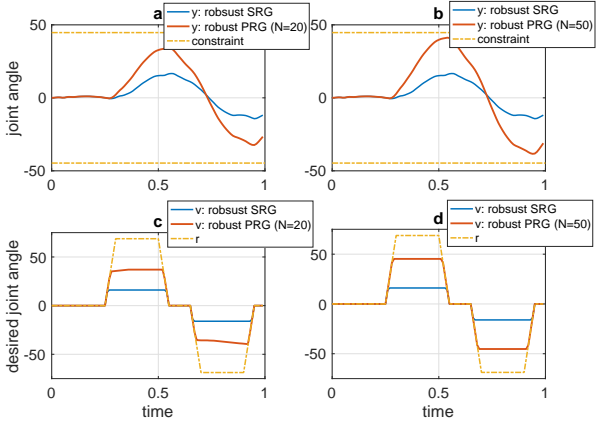}
\caption{Comparison of PRG with  disturbance  preview and SRG with bounded unknown disturbance.}
\label{fig: previewRG disturbance}
\end{figure}

 The following observations can be made. First, the output from all methods satisfy the constraints for all time steps, as expected. Second, it can be seen that $v(t)$ given by PRG is closer to $r(t)$ (less conservative) than that of SRG. The reason 
 is that the constraint set for PRG with disturbance preview is less conservative than that for SRG with unknown disturbance. Third, the longer the preview horizon, the better the performance (less conservative) becomes.  

Finally, note that for the case where the future information of both reference and disturbance is available, the methods from Sections \ref{sec: preview RG} and \ref{sec: PRG with prevew disturbance} can be combined to take the preview of reference and disturbance into consideration simultaneously.

\subsection{Robust PRG for systems with Parametric Uncertainties}
Reconsider system $G(z)$ in Figure \ref{fig: preview Governor scheme}, but now with modeling uncertainty on the $A$ and $B$ matrices:
\begin{equation}\label{eq: LTI with uncertainties}
\begin{aligned}
& x(t+1)=A(t) x(t) + B(t) v(t) \\
& y(t)= C x(t) + D v(t)
\end{aligned}
\end{equation}
where the pair $(A(t),B(t))$ is assumed to  belong to an uncertainty polytope defined by the
convex hull of the matrices:
\[ (A(t),B(t))\in  \mbox{conv}\{(A^{(1)},B^{(1)}), \ldots,  (A^{(L)},B^{(L)})\}\]
where $L$ is the number of vertices in the uncertainty polytope. It is shown in  \cite{pluymers2005efficient} that a robust MAS can be constructed for this uncertain system. We note that a similar idea can be implemented for the PRG. Similar to  \eqref{augment state-space system}, we write the dynamics in terms of the lifted input, but now we consider the  uncertainties:
\begin{equation}\label{augment state-space system with uncertainty}
\begin{aligned}
 &   x(t+1) = A(t) x(t) + 
    \underbrace{\begin{bmatrix}
    \setlength\arraycolsep{1pt}
    B(t) & 0& \ldots&0
    \end{bmatrix}}_{\widetilde{B}(t)} v_N(t),\\
&    y(t) = C x(t) + 
    \widetilde{D}v_N(t)
\end{aligned}
\end{equation}
Then, the pair $(A(t),\widetilde{B}(t))$ must belong to a convex hull of the following matrices:
\[ (A(t),\widetilde{B}(t))\in \mbox{conv}\{(A^{(1)},\widetilde{B}^{(1)}),\ldots (A^{(L)},\widetilde{B}^{(L)})\}\]
where  $\widetilde{B}^{(l)}= \begin{bmatrix}
    B^{(l)} & 0& \ldots&0
    \end{bmatrix}, \quad l=1,\ldots,L.$
By doing this, the method in \cite{pluymers2005efficient} can be used to compute the robust MAS for PRG with system uncertainties. For the sake of brevity, we will not provide numerical examples in this section.

\begin{prop}\label{prop: robustPRG to parametric}
The robust PRG formulation described above is  BIBO stable, recursively feasible, and for a constant $r$, $v(t)$ converges.
\end{prop}
\begin{proof}
As proved in \cite{pluymers2005efficient}, the robust MAS for PRG is  positively invariant. The results follow using a similar approach as the proof of Proposition~\ref{prop:1}.
\end{proof}

\subsection{PRG for Systems with Uncertain Preview Information}
In previous sections, we assumed that the preview information is accurate along the preview horizon. However, this assumption might not hold in practice since the preview information may come from inaccurate measurements or uncertain models. In this section, we present an extension of PRG that can handle inaccurate preview information of references. As can be seen from \eqref{eq: LP to compute kappa}, if $r_N(t)$ has incorrect values, then  $v_N(t)$ will be calculated incorrectly. The implication of this is that, in the next time step, $v_N(t+1)$ will be computed based on the delayed version of this incorrect $v_N(t)$, which, for example, might cause $v(t)$ to change before $r(t)$ does. This behavior is unacceptable for most systems. Note that the constraints will still be satisfied; however, as argued above, performance may suffer. Thus, our solution in this section focuses on avoiding this loss of performance when the preview information is inaccurate.  
Essentially, we achieve this goal by modifying the update law for $v_N(t)$. 

To begin, we assume that at time $t$,  $r(t)$ is accurately known, but there is uncertainty on the value of $r$ along the preview horizon, i.e., $r(t+1),\ldots,r(t+N)$ are inaccurate. 
To accommodate this uncertainty, we modify \eqref{eq: A_bar_def} to:
\begin{equation}\label{eq: A_bar_wrong_preivew}
    \bar{A}=\begin{bmatrix}
    1-\lambda_1 & \lambda_1 &0 &\ldots &0\\
    1-\lambda_1 & \lambda_1-\lambda_2  &\lambda_2 &\ldots &0 \\
    \vdots &\vdots  &\vdots &\vdots & \vdots\\
    1-\lambda_1 & \lambda_1-\lambda_2  &\lambda_2-\lambda_3 & \ldots &\lambda_{N} \\
        1-\lambda_1 & \lambda_1-\lambda_2  &\lambda_2-\lambda_3 & \ldots &\lambda_{N} 
    \end{bmatrix}
\end{equation}
where $\lambda_i \in [0,1]$ are tuning parameters to account for preview uncertainties. This means that instead of the delayed structure of \eqref{eq: A_bar_def}, $v_N(t)$ now evolves such that the value of $v$ at each time along the preview window is a convex combination of the value at that time and the values in the previous times. This effectively incorporates ``forgetting" into the formulation to counteract the uncertainties in preview information.  
We tune $\lambda_i$ according to the relative level of uncertainty on the $i$-th preview information, $r(i)$. Specifically, if the uncertainty is large, $\lambda_i$ is chosen to be close to $0$.  If this is the case, PRG with the new $\bar{A}$ matrix (i.e.,~\eqref{eq: A_bar_wrong_preivew}) will turn to SRG. If $r_i$ is very accurate, on the other hand, $\lambda_i$ will be chosen close to $1$. Then, PRG with the new $\bar{A}$ will turn to standard PRG (as shown in Section \ref{sec: preview RG}). Thus, robust PRG is a proper extension of both standard PRG and SRG. Note that typically the level of uncertainty on $r_i$ increases along the preview horizon, which means that the sequence of $\lambda_i$ is increasing. 

The construction of $O_\infty^N$ and the computation of $v_N$ are the same as \eqref{eq:O_inf for preview} and \eqref{eq: kappa for multi kappa}, except that  $\bar{A}$ is changed  from \eqref{eq: A_bar_def} to \eqref{eq: A_bar_wrong_preivew}.  Now, we will  illustrate this approach using the one-link arm robot example (see the example in Section  \ref{sec: preview RG}). Suppose first that the preview horizon, $N$, is equal to $1$. The slice of $O_\infty^N$ at $x=0$ is shown in Figure \ref{fig: slice of MAS}. The red region corresponds to the slice of $O_\infty^N$ given by the standard PRG at $x=0$. The green region represents the slice of $O_\infty^N$ given by the robust PRG at $x=0$, where we have chosen $\lambda_1=0.2$. Note that in this case ($N=1$), the matrix $\bar{A}$ has only one $\lambda$. 
\begin{figure}[t]
\centering	
\includegraphics[scale=0.28]{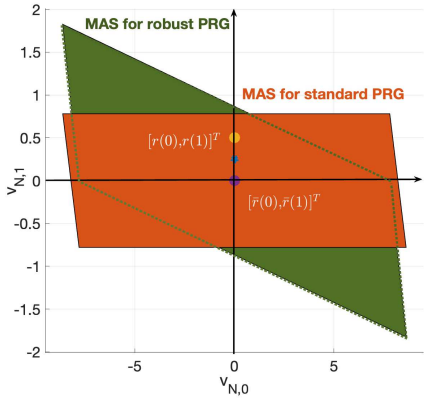}
\caption{The slice of $O_\infty^N$ at $x=0$. $v_{N,0}$ and $v_{N,1}$ represent the first and second element in $v_N$, respectively.}
\label{fig: slice of MAS}
\end{figure}
The situation we consider is as follows. Suppose at $t=0$, the preview information is given by $r_N(0):=(r(0),r(1))=(0,0.5)$ (yellow dot in Figure \ref{fig: slice of MAS}). However, assume that the {\it actual} preview information is $(\bar{r}(0),\bar{r}(1))=(0,0)$ (purple dot), which is unknown to the PRG. Note that $r(0)=\bar{r}(0)$ since we assume that the current information is accurate. Obviously, $v_N(0)$'s given by standard PRG and robust PRG (discussed in previous paragraph) will  be equal to $r_N(0)$ (i.e., $v_N(0)=(0,0.5)$) since $r_N(0)$ is in both MASs (red and green regions in Figure \ref{fig: slice of MAS}). 
In the next time step ($t=1$), if $\kappa \neq 1$,  $v_N(1)$ given by standard PRG (i.e., \eqref{eq: LP to compute kappa}) is a convex combination between the delay version of $v_N(0)$ (i.e., $(0.5,0.5)$) and $r_N(1)$. For robust PRG, $v_N(1)$ is a convex combination between $\bar{A}v_N(0)=(0.45,0.45)$ and $r_N(1)$. Note that if the uncertainty is large, $\lambda_1$ would  be chosen  close to $1$ and $\bar{A}v_N(0)$ will be close to $(0,0)$. By doing so,  robust PRG decreases the impact of the wrong previous information on the computation of current $v_N$ by multiplying the previous information with  values smaller than $1$ (i.e., $\lambda$'s).   

\begin{remark}
By comparing \eqref{eq: A_bar_wrong_preivew} and \eqref{eq: A_bar_def}, it can be seen that, for robust PRG, the last element in $v_N$  is multiplied by a value less than $1$ and the first element in $v_N$ is multiplied by a value that is not equal to $0$, which is different from those in standard PRG. This implies that the MAS for robust PRG will stretch along the $v_{N,N}$ axis and shrink along the $v_{N,0}$ axis  as compared with the MAS for standard PRG (as shown in Figure \ref{fig: slice of MAS}). Thus, unlike the robust MAS for disturbances and polytopic uncertainties, it is not the case   that the robust MAS for inaccurate preview information is a subset of the MAS for standard PRG. 
\end{remark}

\begin{figure}[t]
\centering	
\includegraphics[scale=0.28]{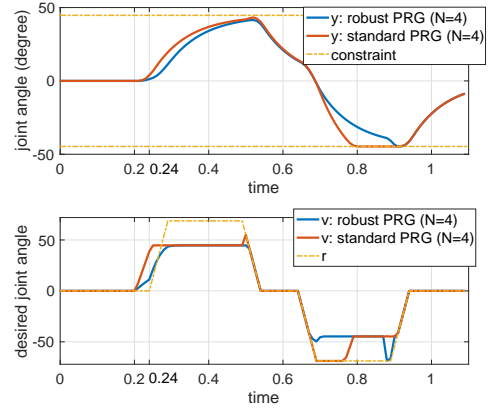}
\caption{Comparison of standard PRG and robust PRG. Red lines represent the simulation results of standard PRG and blue lines refer to the simulation results of robust PRG.}
\label{fig: convex PRG}
\end{figure}

We now perform numerical simulations of the one-link arm robot example. We change the preview horizon to $N=4$ for the sake of illustration.  We consider the scenario where the assumed preview information is larger than the actual preview information along the preview horizon. The $\bar{A}$ is chosen to be:
\begin{equation}\label{eq: A_bar_example}
\bar{A}=\begin{bmatrix}
0.1 & 0.9 &0 &0 &0 \\
0.1 & 0.15 &0.75 &0&0\\
0.1 & 0.15 &0.3 &0.45&0\\ 
0.1 & 0.15 &0.3 &0.35 &0.1\\
0.1 & 0.15 &0.3 &0.35 &0.1
\end{bmatrix}
\end{equation}
Note that several sets of $\lambda$'s were tried and the above $\bar{A}$ was found to result in the best performance (will be explained later). We acknowledge that there are other possibilities to choose $\bar{A}$ and finding the optimal set of $\lambda$s will be explored in future work. The comparison between standard PRG and robust PRG is shown in Figure \ref{fig: convex PRG}. It can be seen that when $t \in~ [0.21,0.24]$, $v(t)$ given by robust PRG (blue line) is closer to $r(t)$ than that given by standard PRG (red line).  The reason for the behavior can be explained as follows.  At $t=0.20$,  $v_N(20)$'s given by standard PRG and robust PRG are both equal to $r_N(20)$. Recall that the first element in $r_N(20)$ is accurate and the rest elements in $r_N(20)$ are inaccurate. Then, in the next time step, for standard PRG, $v(21)$ is calculated as a convex combination between the delayed version of $r_N(20)$ (inaccurate) and $r_N(21)$. With $\kappa=0.25$,  $v(21)$ is calculated as $0.13$. For robust PRG,  $v(21)$ is computed as a convex combination between  $\bar{A}r_N(20)$, where $\bar{A}$ is shown in \eqref{eq: A_bar_example}, and $r_N(20)$. With  $\kappa=0.69$, $v(21)=0.05$. Note that by increasing the value of $\lambda_1$ in \eqref{eq: A_bar_example} (i.e., $0.1$), $v(21)$ would become  closer to $r(21):=0$. However, the tracking performance would not be improved significantly as compared with SRG.

\begin{prop}\label{prop: robustPRG to inaccurate preview}
The robust PRG formulation to inaccurate preview information is BIBO stable, recursively feasible, and for a constant $r$, $v(t)$ converges.
\end{prop}

\begin{proof}
The proof of BIBO stable and recursively feasibility are the same as the proof for Proposition \ref{prop:1}.  To prove the convergence property, assume $r(t) = r, \forall t\in \mathbb{Z}_+$. Note that $\lim_{j\to\infty} \bar{A}^j=~\bar{A}_0$, where $\bar{A}_0$ is a static matrix, since $\bar{A}$ shown in \eqref{eq: A_bar_wrong_preivew} is a stochastic matrix. Then, from \eqref{eq:kapa}, it can be shown that every element in $v_N(t)$ is monotonic increasing after $\bar{A}^j$ converges and bounded by $r$. Thus, $v_N(t)$ must converge to a limit.
\end{proof}

\section{PRG for multi-input system}\label{sec: PRG for multi-inputs}
In this section, we will briefly introduce an extension of PRG to multi-input systems. One possible, straightforward solution is to apply the PRG ideas from above directly to multi-input systems. However, as will be shown in an example later, this approach might lead to a conservative response since PRG uses a single decision variable to simultaneously govern all the channels of the multi-input system.  To address this shortcoming, we propose another solution, which combines the PRG idea with the Decoupled Reference Governor (DRG) scheme \cite{liu2018decoupled}. Detailed information  will be introduced below. 

To begin, suppose that system $G(z)$ in Figure \ref{fig: preview Governor scheme} has $m$ inputs, i.e., $v(t), r(t)\in\mathbb{R}^m$. Let us denote the preview horizon for the $m$ different inputs (i.e., $r_1, \ldots, r_m$) by $N_1, \ldots, N_m$, respectively. Define the lifted signals $r_N $ and $v_N$ as follows:
\begin{equation}\label{eq: def of v_bar_multi_inputs}
\begin{aligned}
\setlength\arraycolsep{2pt}
r_N(t)=(r_1(t), \hdots, r_1(t+N_1),
\hdots,
r_m(t), \hdots, r_m(t+N_m))\\
\setlength\arraycolsep{2pt}
v_N(t)=(
v_1(t), \hdots, v_1(t+N_1), \hdots, 
v_m(t), \hdots, v_m(t+N_m))
\end{aligned}
\end{equation}
We can select the dynamics of $v_N(t)$ to be the same as \eqref{eq: A_bar} but with $\bar{A}$  constructed as:
\begin{equation}\label{eq: A_bar_multi_inputs}
\renewcommand*{\arraystretch}{.1}
\bar{A}=\begin{bmatrix}
\bar{I}_{N_1} & \ldots & 0_{N_1\times N_m}\\
\vdots & \ddots &\vdots\\
0_{N_m \times N_1} & \ldots &\bar{I}_{N_m}
\end{bmatrix}
\end{equation}
where $\bar{I}_{N_i}$ ($i=1,\ldots,m$) is defined the same as \eqref{eq: A_bar}, with $N$ replaced by $N_i$.

The construction of $O_\infty^N$ and the calculation of $v_N$ are the same as \eqref{eq:O_inf for preview} and \eqref{eq: LP to compute kappa}, except that $\bar{A}$ is modified to \eqref{eq: A_bar_multi_inputs}.  Below, we will provide an illustrative example to show that this approach works as required but might lead to a  conservative response. Consider a two-link arm robot, which has dynamics as follows \cite{mahil2016modeling}: 
\begin{equation}
\begin{aligned}
&\begin{bmatrix}
\dot{x_1}\\\dot{x_2}\\\dot{x_3}\\\dot{x_4}
\end{bmatrix}=\begin{bmatrix}
0 & 0& 1 &0\\
0 & 0& 0 &1\\
-0.46 & -0.62 & 0 & 0\\
0.25 & -6.62 & 0 & 0\\
\end{bmatrix}\begin{bmatrix}
x_1\\x_2\\x_3\\x_4
\end{bmatrix}+\begin{bmatrix}
0 & 0\\
0 & 0\\
0.78 & -0.04\\
0.04 & 0.13
\end{bmatrix}
\begin{bmatrix}
\tau_1\\\tau_2
\end{bmatrix}
\end{aligned}
\end{equation}
where $\tau_1$ and $x_1:=\theta_1$  are the external torque and the joint angle for the first link, respectively.  Similarly, $\tau_2$ and $x_2:=\theta_2$  are the torque and the joint angle for the second link, respectively. The constrained outputs are $\theta_1$ and $\theta_2$. The constraints on  both joint angles are $[-60,60]$. To implement PRG, the system is first discretized at $T_s=0.01s$.  Then, a state feedback controller is designed to ensure that $\theta_1$ and $\theta_2$ track desired joint angles, $v_1$ and $v_2$, respectively, that is:
{\small
$$
\begin{bmatrix}
\tau_1 \\ \tau_2
\end{bmatrix}=
\begin{bmatrix}
769.23  & 0\\ 
0 & 3333.3
\end{bmatrix}
\begin{bmatrix}
v_1\\
v_2
\end{bmatrix}-\begin{bmatrix}
750	& 155 &	59 &	19\\
-226 &	2867&	-18 &	350
\end{bmatrix}\begin{bmatrix}
x_1\\x_2\\x_3\\x_4
\end{bmatrix}
$$}\par
\noindent The preview horizons for both references are chosen to be $40$ (i.e., $N_1=N_2=40$). The simulation results of PRG of this multi-input system are shown in Figure \ref{fig: prg multi inputs 1}. As can be seen, the outputs both  satisfy the constraints, as required. However, when $t \in [0.8,1.2]$, $v_2$ can not reach $r_2$, which is shown by the purple dashed line in the bottom plot, even though $y_2$ (i.e., red line in the top plot) is far from the lower constraint. This is caused by the fact that $y_1$ (i.e., the blue line in the top plot of Figure \ref{fig: prg multi inputs 1}) reaches the constraint when $t \in [0.8,1.2]$, which results in $\kappa$ being equal to $0$. Since a single $\kappa$ is used in PRG scheme, $v_2$ can not reach $r_2$.

\begin{figure}[t]
\centering	
\includegraphics[scale=0.3]{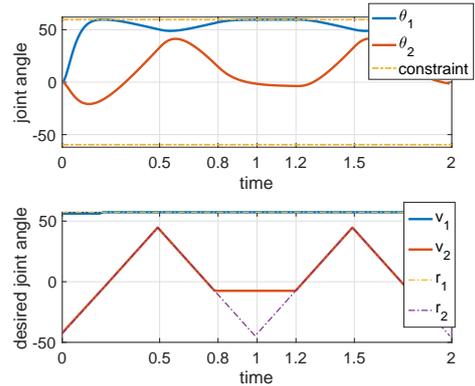}
\caption{Simulation results of PRG for the two-link arm robot. The blue lines represent the response of joint $1$ and the red lines represent the response of joint $2$.}
\label{fig: prg multi inputs 1}
\end{figure}

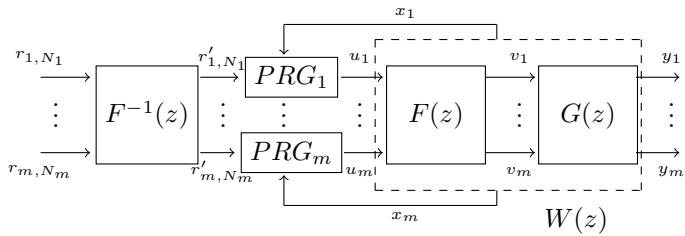
\begin{figure}
\centering
\begin{tikzpicture}
    [L1Node/.style={rectangle,draw=black,minimum size=5mm},
    L2Node/.style={rectangle,draw=black, minimum size=13mm},
    L3Node/.style={rectangle,draw=black, minimum size=30mm}]
      \node[L2Node] (n3) at (1.4, 0){$F^{-1}(z)$};
      \node[L1Node] (n4) at (3.3, 0.5){$PRG_{1}$};
      \draw (3.2,0.1)node [color=black,font=\fontsize{10}{10}\selectfont]{\vdots};
      \draw (2.4,0.1)node [color=black,font=\fontsize{10}{10}\selectfont]{\vdots};
      \draw (4.2,0.1)node [color=black,font=\fontsize{10}{10}\selectfont]{\vdots};
      \draw (6.3,0.1)node [color=black,font=\fontsize{10}{10}\selectfont]{\vdots};
      \draw (8.3,0.1)node [color=black,font=\fontsize{10}{10}\selectfont]{\vdots};
      \draw (0.2,0.1)node [color=black,font=\fontsize{10}{10}\selectfont]{\vdots};
      \node[L1Node] (n4) at (3.3, -0.5){$PRG_{m}$};
      \node[L2Node] (n5) at (5.2, 0){$F(z)$};
      \node[L2Node] (n6) at (7.2, 0){$G(z)$};
	   \draw[dashed](4.4,1)-- (7.9,1);
      \draw[dashed](4.4,1)-- (4.4,-1);
      \draw[dashed](4.4,-1)-- (7.9,-1);
      \draw[dashed](7.9,-1)-- (7.9,1);
      \draw[->](0,0.5)--(0.65,0.5);
      \draw(6,1)--(6,1.2);
      \draw(6,1.2)--(3.2,1.2);
      \draw[->](3.2,1.2)--(3.2,0.8);
      \draw(6,-1)--(6,-1.2);
      \draw(6,-1.2)--(3.2,-1.2);
      \draw[->](3.2,-1.2)--(3.2,-0.8);

      \draw[->](0,-0.5)--(0.65,-0.5);
      \draw[->](2.1,0.5)--(2.6,0.5);
      \draw[->](2.1,-0.5)--(2.55,-0.5);
      \draw[->](3.95,0.5)--(4.5,0.5);
      \draw[->](3.99,-0.5)--(4.5,-0.5);       
      \draw[->](5.85,0.5)--(6.5,0.5);
      \draw[->](5.85,-0.5)--(6.5,-0.5);
      \draw[->](7.83,0.5)--(8.4,0.5);
      \draw[->](7.83,-0.5)--(8.4,-0.5);
      \draw (0,0.75)node [color=black,font=\fontsize{6}{6}\selectfont]{$r_{1,N_1}$};
      \draw (0,-0.75)node [color=black,font=\fontsize{6}{6}\selectfont]{$r_{m,N_m}$};
      \draw (2.4,0.75)node [color=black,font=\fontsize{6}{6}\selectfont]{$r_{1,N_1}'$};
      \draw (2.4,-0.75)node [color=black,font=\fontsize{6}{6}\selectfont]{$r_{m,N_m}'$};
        \draw (4.2,0.75)node [color=black,font=\fontsize{6}{6}\selectfont]{$u_{1}$};
      \draw (4.2,-0.75)node [color=black,font=\fontsize{6}{6}\selectfont]{$u_{m}$};
      \draw (6.3,0.75)node [color=black,font=\fontsize{6}{6}\selectfont]{$v_{1}$};
      \draw (6.3,-0.75)node [color=black,font=\fontsize{6}{6}\selectfont]{$v_{m}$};
      \draw (8.3,0.75)node [color=black,font=\fontsize{6}{6}\selectfont]{$y_{1}$};
      \draw (8.3,-0.75)node [color=black,font=\fontsize{6}{6}\selectfont]{$y_{m}$};
      \draw (6.5,-1.4)node [color=black,font=\fontsize{10}{10}\selectfont]{\hspace{1cm} $W(z)$};
      \draw (4.8,1.35)node [color=black,font=\fontsize{6}{6}\selectfont]{$x_{1}$};
        \draw (4.8,-1.35)node [color=black,font=\fontsize{6}{6}\selectfont]{$x_{m}$}; 
\end{tikzpicture}

\caption{PRG block diagram for square MIMO systems. $r_{i,N_i}(t)$ represents the lifted $r_i$ over the preview horizon, i.e., $r_{i,N_i}(t)=(r_i(t),\ldots,r_i(t+N_i))$.} \label{fig:Decoupled with RG}
 \label{fig: PRG system block for mimo systems}
\end{figure}

From the above discussion, it can be seen that this approach might lead to a conservative response. To address this shortcoming, we propose another method, which  combines the PRG theory with the DRG scheme in \cite{liu2020decoupled}. As a quick review, DRG is based on decoupling the input-output dynamics of the system, followed by the application of SRG to each decoupled channel. In our solution, instead of using SRGs to govern the decoupled channels, we use the PRG presented in Section \ref{sec: preview RG}. This leads to the block diagram shown in Figure \ref{fig: PRG system block for mimo systems}. For ease of presentation, we will assume the system is square (i.e., equal number of inputs and outputs). Non-square systems can be handled as well by considering the DRG scheme for non-square systems presented in \cite{liu2020decoupled}. 

Thus, we suppose the closed-loop system $G(z)$ is described by a $m \times m$ transfer function system matrix  with elements $G_{ij}$.
As shown in Figure \ref{fig: PRG system block for mimo systems}, we decouple $G(z)$ by introducing $F(z)$, leading to a diagonal system: $W(z)$. As discussed in \cite{liu2018decoupled}, there are several ways to construct $W(z)$. For the sake of brevity, we only consider the case where $W(z)$ is given by: $W(z)=\mathrm{diag}(G_{11},G_{22},\ldots,G_{mm})$. 
Note that the same process to combine PRG idea and DRG scheme can be applied to other formulations of $W(z)$ as well. Then, $m$ different PRGs for single-input systems (see Section \ref{sec: preview RG}) are implemented, one for each $W_{ii}$. Finally, as discussed in \cite{liu2018decoupled}, $F^{-1}$ is introduced to ensure that $v$ and $r$ are close if $r$ is not constraint-admissible. Simulation results of the second approach on the two-link arm robot are shown in Figure \ref{fig: prg multi inputs 2}. As can be seen, the outputs  satisfy the constraints, as required.  Also,  when $t \in [0.8,1.2]$, $v_2$  reaches $r_2$, which is represented by the purple dashed line in the bottom plot. This behavior can be explained as follows. When $t=0.8$, the PRGs have the future information that $r_2$ will drop down to $-56$ at $t=1$ and then go up towards $-20$ at $t=1.2$ (recall that $N_1=N_2=40$). Also, from the top plot of Figure \ref{fig: prg multi inputs 2}, it can be seen that when $t \in  [0.8,1.2]$, $y_2$ (red line in the top plot) does not reach the lower constraint, which leads to $\kappa_2=1$ (recall that $\kappa_2$ refers to the optimization variable given by the second PRG). Hence, $v_2$ reaches $r_2$.

\begin{figure}[t]
\centering	
\includegraphics[scale=0.28]{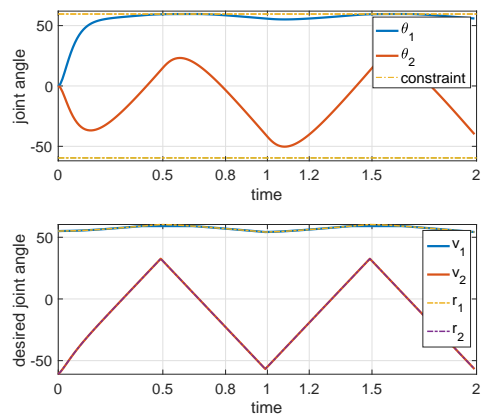}
\caption{Simulation results of PRG for the two-link arm robot. The blue lines represent the response of joint $1$ and the red lines represent the response of joint $2$.
}
\label{fig: prg multi inputs 2}
\end{figure}

Note that both approaches discussed above can be easily extended to multi-$N$ PRG by combining DRG scheme and Multi-N PRG idea. 

\vspace{-0.2cm}

\section{CONCLUSIONS AND FUTURE WORKS }\label{sec:conclusion}

In this work, a reference governor-based method for constraint management of linear systems is proposed. The method is referred to as Preview Reference Governor (PRG) and can systematically account for the preview information of reference and disturbance signals. The method is based on lifting the input of the system to a space of higher dimension and designing maximal admissible sets based on the system with lifted input. We showed a limitation of PRG and proposed an alternative method, which we refered to as Multi-horizon PRG (multi-$N$ PRG), to overcome the limitation. Disturbance previews, parametric uncertainties, and inaccurate preview reference information were also addressed. We also showed that the PRG for multi-input systems using the lifting idea (i.e., first solution in Section \ref{sec: PRG for multi-inputs}) might cause conservative response. Thus, we proposed another method, which combines the Decoupled Reference Governor scheme and PRG, to overcome this limitation.

Future work will explore preview control in the context of Vector Reference Governors, as well as finding the optimal set of $\lambda$s that gives the best performance in our robust PRG formulation.  We will also investigate the extension of PRG to nonlinear systems.

\vspace{-0.4cm}

\medskip

\bibliography{mybibfile}{}
\bibliographystyle{plain}

\end{document}